\documentclass[11pt,english,pra]{article}
\usepackage[T1]{fontenc}
\usepackage[latin9]{inputenc}
\usepackage[letterpaper]{geometry}
\usepackage{amsthm}
\usepackage{amsmath}
\usepackage{amssymb}

\makeatletter
\theoremstyle{plain}

\makeatletter
\@ifundefined{textcolor}{}
{%
 \definecolor{BLACK}{gray}{0}
 \definecolor{WHITE}{gray}{1}
 \definecolor{RED}{rgb}{1,0,0}
 \definecolor{GREEN}{rgb}{0,1,0}
 \definecolor{BLUE}{rgb}{0,0,1}
 \definecolor{CYAN}{cmyk}{1,0,0,0}
 \definecolor{MAGENTA}{cmyk}{0,1,0,0}
 \definecolor{YELLOW}{cmyk}{0,0,1,0}
 }

\usepackage{amsthm}

\makeatletter
\@ifundefined{textcolor}{}
{%
 \definecolor{BLACK}{gray}{0}
 \definecolor{WHITE}{gray}{1}
 \definecolor{RED}{rgb}{1,0,0}
 \definecolor{GREEN}{rgb}{0,1,0}
 \definecolor{BLUE}{rgb}{0,0,1}
 \definecolor{CYAN}{cmyk}{1,0,0,0}
 \definecolor{MAGENTA}{cmyk}{0,1,0,0}
 \definecolor{YELLOW}{cmyk}{0,0,1,0}
 }
\theoremstyle{plain}

\makeatletter
\@ifundefined{textcolor}{}
{%
 \definecolor{BLACK}{gray}{0}
 \definecolor{WHITE}{gray}{1}
 \definecolor{RED}{rgb}{1,0,0}
 \definecolor{GREEN}{rgb}{0,1,0}
 \definecolor{BLUE}{rgb}{0,0,1}
 \definecolor{CYAN}{cmyk}{1,0,0,0}
 \definecolor{MAGENTA}{cmyk}{0,1,0,0}
 \definecolor{YELLOW}{cmyk}{0,0,1,0}
 }

\makeatletter
\@ifundefined{textcolor}{}
{%
 \definecolor{BLACK}{gray}{0}
 \definecolor{WHITE}{gray}{1}
 \definecolor{RED}{rgb}{1,0,0}
 \definecolor{GREEN}{rgb}{0,1,0}
 \definecolor{BLUE}{rgb}{0,0,1}
 \definecolor{CYAN}{cmyk}{1,0,0,0}
 \definecolor{MAGENTA}{cmyk}{0,1,0,0}
 \definecolor{YELLOW}{cmyk}{0,0,1,0}
 }

\makeatother

\newtheorem{thm}{Theorem}  
 \newtheorem{lem}{Lemma}
\newtheorem{prop}{Proposition}  

\makeatother

\makeatother

\makeatother

\makeatother

\makeatother

\makeatother

\usepackage{babel}

\begin{document}

\title{Entanglement, mixedness and perfect local discrimination of orthogonal
quantum states}

\author{Somshubhro Bandyopadhyay%
\thanks{Department of Physics and Center for Astroparticle Physics and Space
Science, Block EN, Sector V, Salt Lake, Kolkata 700091; email: som@bosemain.boseinst.ac.in%
}}
\maketitle
\begin{abstract}
It is shown that local distinguishability of orthogonal mixed states
can be completely characterized by local distinguishability of their
supports irrespective of entanglement and mixedness of the states.
This leads to two kinds of upper bounds on the number of perfectly
locally distinguishable orthogonal mixed states. The first one depends
only on pure-state entanglement within the supports of the states,
and therefore may be easy to compute in many instances. The second
bound is optimal in the sense that it optimizes the bounding quantities,
not necessarily function of entanglement alone, over all orthogonal
mixed state ensembles (satisfying certain conditions) admissible within
the supports of the density matrices. 
\end{abstract}

\section{Introduction}

A characteristic feature of quantum theory is that composite quantum
systems, whose parts do not interact, may possess nonlocal properties.
For example, entanglement \cite{Horodeckis-2009} and quantum information
both exhibit nonlocality. Entangled states are nonlocal because they
give rise to correlations that cannot be explained by local hidden
variable theories \cite{Bell-1964}, whereas, nonlocality of quantum
information is in the sense that a measurement on the whole system
sometimes reveals more information about the state than coordinated
local measurements on its parts \cite{Bennett-I-99,Bennett-II-99 +Divin-2003,Peres-Wootters-1991,Ghosh-2001,Horodecki-2003,Bandyo-2011}.
This nonlocal nature of quantum information is generally manifested
in the setting of local discrimination of quantum states \cite{Bandyo+Walgate-2009,Bandyo-2011,Bennett-I-99,Chefles2004,Ghosh-2001,Ghosh-2002,Hayashi-etal-2006,Horodecki-2003,Nathanson-2005,Peres-Wootters-1991,Virmanietal2001,Walgate-2000-2002,Watrous-2005,Duan2007-2009,Hayashi-2}.
One of the principal goals in quantum information theory is to understand
and quantify the relationship between entanglement and nonlocality
of quantum information. 

The difficulty in quantifying the role of entanglement in local state
discrimination is evident from some of the early results, which show
that the presence of entanglement is neither necessary nor sufficient
to ensure whether a given set of orthogonal states is locally indistinguishable.
That entanglement is not necessary is evident from the examples of
locally indistinguishable sets of orthogonal product states exhibiting
{}``non-locality without entanglement'' or forming an unextendible
product basis (UPB)\cite{Bennett-I-99,Bennett-II-99 +Divin-2003}.
On the other hand, any two orthogonal states can be perfectly distinguished
no matter how entangled they are \cite{Walgate-2000-2002}, showing
that entanglement is not sufficient for local indistinguishability.
Nevertheless, entanglement is often the key factor in a typical locally
indistinguishable set as in the case of a complete bipartite orthogonal
basis containing one or more entangled states \cite{Horodecki-2003,Ghosh-2001};
if such a set can be perfectly distinguished locally, then one can
create entanglement from product states using LOCC \cite{Horodecki-2003,Ghosh-2001},
a task known to be impossible.

Significant progress, which also motivated the present work, was reported
in \cite{Hayashi-etal-2006}, where it was shown that entanglement
does guarantee difficulty in local state discrimination. In particular,
it was shown that if the states (pure or mixed) $\sigma_{1},\sigma_{2},...,\sigma_{N}$
can be perfectly discriminated by LOCC, then the number of states
is bounded by\begin{equation}
N\leq D/\overline{d\left(\sigma_{i}\right)}\leq D/\overline{r(\sigma_{i})}\leq D/\overline{2^{E\left(\sigma_{i}\right)}}\leq D/\overline{2^{G(\sigma_{i})}},\label{eq:Hayashi}\end{equation}
where, $D$ is the dimension of the Hilbert space of the composite
quantum system; $d\left(\sigma\right)$ is a quantity (to be defined
later) resembling distance to the nearest separable state; $r\left(\sigma\right)=\mbox{rank}\left(\rho\right)\left(1+R_{g}\left(\rho\right)\right)$,
where $\rho$ is the normalized projector onto the support \cite{support}
of $\sigma$ and $R_{g}\left(\rho\right)$ is the robustness of entanglement
\cite{Harrow-3-Steiner-Vidal}; $E\left(\sigma_{i}\right)=E_{R}(\sigma_{i})+S\left(\sigma_{i}\right)$,
where $E_{R}\left(\sigma\right)$ is the relative entropy \cite{Vedral-Plenio-1998}
and $S\left(\sigma\right)$ is the von Neumann entropy; $G\left(\sigma\right)$
is the geometric measure \cite{Hayashi-etal-2006,shimoy-geometric-measure},
and $\overline{x_{i}}=\frac{1}{N}\sum_{i=1}^{N}x_{i}$ denotes the
average. If the inequality is violated for any of the bounding quantities,
then we can certainly conclude that the given set of states cannot
be perfectly locally distinguished. However, if the inequality is
satisfied then no such definite conclusion can be drawn. Applications
of inequality (\ref{eq:Hayashi}) for LOCC discrimination of interesting
multipartite ensembles having certain group symmetries can be found
in \cite{Hayashi-2}.

For pure states the bounding quantities (from right to left) correspond
to well-defined distance-like entanglement measures, namely, geometric
measure, relative entropy, and robustness of entanglement, thereby
allowing a clear interpretation: the number of pure states that can
be perfectly distinguished by LOCC is bounded by the total dimension
over average entanglement. This therefore clarifies the matter to
a great extent for pure states. For mixed states, however, no such
clear conclusion can be drawn, and the role of entanglement still
remains unclear. 

The purpose of the present work is to investigate how local distinguishability
of a given set of orthogonal mixed states depend on entanglement and
mixedness of the states. We first show that local distinguishability
of mixed states can be completely characterized by local distinguishability
of their supports%
\footnote{The support of a density matrix is the subspace spanned by its eigenvectors
corresponding to the nonzero eigenvalues.%
}. In particular, we establish a simple equivalence between local discrimination
of orthogonal states and subspaces in the sense that a given set of
density matrices can be perfectly distinguished by LOCC if and only
if their supports are also perfectly locally distinguishable, and
moreover, if the states can be perfectly distinguished, then the separable
measurement that distinguishes the states also distinguishes the supports
and vice versa. We use this fact to obtain the following results:
(a) state-specific properties such as inseparability and mixedness
of the density matrices (whose local distinguishability is under consideration)
do not have any special role in determining their local distinguishability,
(b) local distinguishability of mixed states may be completely determined
by maximal pure state entanglement within their supports, and (c)
the number of LOCC distinguishable orthogonal mixed states can be
bounded by the quantities that are optimized over all orthogonal mixed
state ensembles having identical supports. We now briefly discuss
results (a)-(c). 

For result (a), we show that the state-specific properties such as
inseparability and mixedness of the given mixed states do not have
any fundamental role in determining their local distinguishability.
To see this, suppose $\mathbb{S}=\left\{ \sigma_{1},\sigma_{2},...,\sigma_{N}\right\} $
is a set of orthogonal density matrices whose local distinguishability
is under question. Let $\left\{ s_{1},s_{2},...,s_{N}\right\} $ be
the set of orthogonal subspaces where, $s_{i}$ is the support of
$\sigma_{i}$. It is clear that infinitely many sets like $\mathbb{S}$
exist, where each set contains $N$ orthogonal density matrices with
the property that $s_{i}$ is the support of its $i^{\mbox{th}}$
element. Let $\mathbb{Q}$ be the collection of all such sets; that
is, $\mathbb{Q}=\left\{ \mathbb{S},\mathbb{S}^{\prime},...\right\} $.
We then show that either every set in $\mathbb{Q}$ is perfectly distinguishable
by LOCC or none of them are, regardless of how entangled or mixed
the states in a given set are. That is, if $\mathbb{S}$ is perfectly
distinguishable by LOCC, then so is any set, say $\mathbb{S}^{\prime}\in\mathbb{Q}$
and vice versa, even though average entanglement or mixedness could
be very different for the states in $\mathbb{S}$ and $\mathbb{S}^{\prime}$
(for instance, the density matrices in $\mathbb{S}$ may be highly
entangled, whereas the density matrices in $\mathbb{S}^{\prime}$
may be very weakly entangled). We call this property subspace degeneracy.
Thus the state-specific properties of the density matrices $\sigma_{i}$
do not have any special role as far as their local distinguishability
is concerned. 

For results (b) and (c) we use the above observations to present upper
bounds on the number of perfectly locally distinguishable orthogonal
mixed states. In particular, we obtain two kinds of upper bounds.
The first one shows that the number of orthogonal density matrices
that can be perfectly distinguished locally is bounded above by the
total dimension over the average of maximal pure state entanglement
in the supports of the density matrices. This bound is not necessarily
optimal but depends only on pure-state entanglement within the supports
of the states and therefore may be easy to compute in many instances.
This shows that local distinguishability of mixed states may be determined
by pure-state entanglement alone. The second bound is optimal in the
sense that it optimizes the bounding quantities over all orthogonal
ensembles (satisfying certain conditions) admissible within the supports
of the density matrices.

\section{Necessary conditions for perfect LOCC state discrimination}

Let $\mathcal{H}$ be the Hilbert space of a composite quantum system
and $D=\dim\mathcal{H}$. Throughout this paper we consider only finite-dimensional
systems. We note that any measurement realized by LOCC is separable
(the converse is not true \cite{Bennett-I-99}). A separable measurement
$\Pi=\left\{ \Pi_{1},\Pi_{2},\cdots,\Pi_{n}\right\} $ on $\mathcal{H}$
is a POVM satisfying $\sum_{i=1}^{n}\Pi_{i}=\mathcal{I_{H}}$, where
$\Pi_{i}$ is a separable, positive semi-definite operator for every
$i$. Therefore, if a set of quantum states is perfectly distinguishable
by LOCC, then there exists a separable measurement distinguishing
the states. For a necessary and sufficient condition for perfect discrimination
by separable measurements, see \cite{Duan2007-2009}.

We now state two necessary conditions for perfect LOCC state discrimination.
The first condition and its variants can be found in Refs. \cite{Hayashi-etal-2006,Watrous-2005,Chefles2004,Nathanson-2005,Duan2007-2009}
and the second condition is due to Ref. \cite{Hayashi-etal-2006}. 

\begin{prop} If the orthogonal quantum states $\sigma_{1},\sigma_{2},...,\sigma_{N}$
are perfectly distinguishable by LOCC then it is necessary that there
exists a separable POVM $\Pi=\left\{ \Pi_{1},\Pi_{2},\cdots,\Pi_{N}\right\} $
such that \begin{equation}
\mbox{Tr}\left(\Pi_{i}\sigma_{j}\right)=\delta_{ij}.\label{prop1-eq-1}\end{equation}
\end{prop}

\begin{prop}

A necessary condition for perfect LOCC discrimination of the states
$\sigma_{1},\sigma_{2},...,\sigma_{N}$ by a separable POVM $\Pi=\left\{ \Pi_{1},\Pi_{2},\cdots,\Pi_{N}\right\} $
is that the following inequality is satisfied:\begin{equation}
\sum_{i=1}^{N}d\left(\sigma_{i}\right)\leq D\label{prop-2-eq-1}\end{equation}
 where, $d\left(\sigma_{i}\right):=\min\frac{\mbox{Tr}\left(\Pi_{i}\right)}{\mbox{Tr}\left(\sigma_{i}\Pi_{i}\right)}$
such that $0\leq\frac{\Pi_{i}}{\mbox{Tr}\left(\sigma_{i}\Pi_{i}\right)}\leq\mathcal{I}$.\end{prop}

Let us remark that the above necessary condition is particularly useful
for bounding the number of states that can be perfectly discriminated
by LOCC \cite{Hayashi-etal-2006}.

\section{Results}

\subsection{Local discrimination of orthogonal subspaces}

We first explain what we mean by LOCC discrimination of orthogonal
subspaces (for discrimination of non-orthogonal subspaces using global
measurements see \cite{Hillery-2006}). In local discrimination of
orthogonal subspaces, a pure quantum state shared between several
observers is guaranteed to belong to a subspace chosen from a known
collection of orthogonal subspaces. The goal is to determine by LOCC
to which subspace the state belongs without making any error. We assume
that within each subspace each state is equally likely, and so are
the subspaces. We will say that the subspaces $\left\{ \mathcal{S}_{1},\mathcal{S}_{2},...,S_{k}\right\} $
are perfectly locally distinguishable if we can perfectly distinguish
the set of density matrices $\left\{ \rho_{1},\rho_{2},...,\rho_{k}\right\} $
by LOCC, where $\rho_{i}$ is the normalized projector onto the subspace
$\mathcal{S}_{i}$. Clearly, the problem of local discrimination of
orthogonal subspaces is a special case of the general problem. We
begin with a simple but useful lemma.

\begin{lem} If the orthogonal subspaces $\left\{ \mathcal{S}_{1},\mathcal{S}_{2},...,\mathcal{S}_{k}\right\} $
are perfectly LOCC distinguishable, then so are the density matrices
$\left\{ \omega_{1},\omega_{2},...,\omega_{k}\right\} $ where, $\omega_{i}\in\mathcal{S}_{i}$.
\end{lem}
\begin{proof}
That the subspaces $\left\{ \mathcal{S}_{1},\mathcal{S}_{2},...,\mathcal{S}_{k}\right\} $
are perfectly LOCC distinguishable means that the set of orthogonal
density matrices $\left\{ \rho_{1},\rho_{2},...,\rho_{k}\right\} $,
the normalized projectors onto the subspaces, can be perfectly distinguished.
Thus there exists a locally implementable separable POVM $\Pi=\left\{ \Pi_{1},\Pi_{2},...,\Pi_{k}\right\} $
such that $\mbox{Tr}\left(\Pi_{i}\rho_{j}\right)=\delta_{ij}$. Denoting
$\rho_{j}=\frac{1}{\dim\mathcal{S}_{j}}\Lambda_{j}$, where, $\Lambda_{j}$
is the projection operator onto $\mathcal{S}_{j}$, we get,\begin{equation}
\frac{1}{\dim\mathcal{S}_{j}}\mbox{Tr}\left(\Pi_{i}\Lambda_{j}\right)=\delta_{ij};\label{lem-1-eq-1}\end{equation}
 thus for $i\neq j$ the POVM elements $\left\{ \Pi_{i}\right\} $
are all orthogonal to the subspace $\mathcal{S}_{j}$. Now for any
density matrix $\Delta$ and a POVM $\mathcal{M}=\left\{ \mathcal{M}_{i}:i=1,...,k\right\} $
the relation \begin{equation}
\sum_{i}\mbox{Tr}\left(\mathcal{M}_{i}\Delta\right)=1,\label{lem-1-eq-1.1}\end{equation}
 is valid. The summation indicates that the sum of the probabilities
must add up to 1 when the measurement $\mathcal{M}$ is performed
on the state $\Delta$, where, $\mbox{Tr}\left(\mathcal{M}_{i}\Delta\right)$
is the probability of obtaining outcome $i$. Suppose the POVM $\Pi$
is implemented on the given state chosen from $\left\{ \omega_{1},\omega_{2},...,\omega_{k}\right\} $.
We therefore have \begin{equation}
\sum_{i}\mbox{Tr}\left(\Pi_{i}\omega_{j}\right)=1,\label{lem1-eq-2}\end{equation}
 where $\mbox{Tr}\left(\Pi_{i}\omega_{j}\right)$ is the probability
of obtaining outcome $i$ when the input state is $\omega_{j}$. Because
$\omega_{j}\in\mathcal{S}_{j}$, $\omega_{j}$ must be orthogonal
to all POVM elements $\Pi_{i};i\neq j$. Therefore, \begin{equation}
\mbox{Tr}\left(\Pi_{i}\omega_{j}\right)=\delta_{ij}.\label{lem-1-eq-3}\end{equation}
 Thus the POVM $\Pi$ also perfectly distinguishes the states $\left\{ \omega_{1},\omega_{2},...,\omega_{k}\right\} $. 
\end{proof}

\subsection{Equivalence of LOCC discrimination of states and subspaces}

For a given set of orthogonal density matrices $\left\{ \rho_{i}:i=1,...,N\right\} $,
consider the set of subspaces $\left\{ \mathcal{S}_{1},...,\mathcal{S}_{N}\right\} $,
where $\mathcal{S}_{i}$ is the support of $\rho_{i}$. The orthogonality
of the density matrices implies that the supports are orthogonal.
We now show that the problems of local discrimination of orthogonal
states and subspaces are equivalent in the following sense. 

\begin{thm}The density matrices $\rho_{1},...,\rho_{N}$ are perfectly
distinguishable by LOCC if and only if their supports are. Moreover,
if the states are perfectly distinguishable, then the measurement
that distinguishes the states also distinguishes their supports and
vice versa. \end{thm}
\begin{proof}
Let $\Pi=\left\{ \Pi_{i}:i=1,...,N\right\} $ be the POVM that perfectly
distinguishes the set of density matrices $\left\{ \rho_{i}:i=1,...,N\right\} $
by LOCC. Therefore, $\mbox{Tr}\left(\Pi_{i}\rho_{j}\right)=\delta_{ij}$.
Let $s_{j}$ be the support of $\rho_{j}$ and $\varrho_{j}=\frac{1}{|\mathcal{P}_{j}|}\mathcal{P}_{j}$,
where $\mathcal{P}_{j}$ is the projector onto the subspace $s_{j}$
and $|\mathcal{P}_{j}|=\dim s_{j}$. To prove that the POVM $\Pi$
also perfectly distinguishes the subspaces $\mathcal{S}_{1},\mathcal{S}_{2},...,\mathcal{S}_{N}$,
we only need to show that $\mbox{Tr}\left(\Pi_{i}\varrho_{j}\right)=\delta_{ij}$.
That the POVM is locally implementable holds by the assumption that
it perfectly distinguishes $\left\{ \rho_{i}\right\} $. Consider
first the diagonal decomposition of $\rho_{j}:$\begin{equation}
\rho_{j}=\sum_{l=1}^{|\mathcal{P}_{j}|}p_{l}^{j}|\phi_{l}^{j}\rangle\langle\phi_{l}^{j}|.\label{thm-1-eq-1}\end{equation}
 From Eq.$\,$(\ref{lem-1-eq-1.1}) for every $l$ we have \begin{equation}
\sum_{i}\mbox{Tr}\left(\Pi_{i}|\phi_{l}^{j}\rangle\langle\phi_{l}^{j}|\right)=1.\label{thm-1-eq-2}\end{equation}
 Also $\mbox{Tr}\left(\Pi_{i}\rho_{j}\right)=\delta_{ij}$ implies
that all POVM elements $\Pi_{i}\;\left(i\neq j\right)$ are orthogonal
to the states $\left\{ |\phi_{l}^{j}\rangle\;:l=1,...,d_{j}\right\} $.
Using this fact, the above equation reduces to, \begin{equation}
\mbox{Tr}\left(\Pi_{i}|\phi_{l}^{j}\rangle\langle\phi_{l}^{j}|\right)=\delta_{ij}\quad:\forall l.\label{thm-1-eq-3}\end{equation}
 Noting that $\varrho_{j},$ the normalized projector onto the subspace
$s_{j}$ can be written as\begin{equation}
\varrho_{j}=\frac{1}{|\mathcal{P}_{j}|}\sum_{l=1}^{|\mathcal{P}_{j}|}|\phi_{l}^{j}\rangle\langle\phi_{l}^{j}|.\label{thm-1-eq-4}\end{equation}
 we immediately obtain $\mbox{Tr}\left(\Pi_{i}\varrho_{j}\right)=\delta_{ij}$
using Eqs.$\,$(\ref{thm-1-eq-3}) and (\ref{thm-1-eq-4}). Thus the
subspaces can be perfectly distinguished and the POVM $\Pi$ distinguishes
them. The rest of the proof, namely, the POVM that perfectly distinguishes
the orthogonal subspaces $\left\{ \mathcal{S}_{1},\mathcal{S}_{2},...,\mathcal{S}_{N}\right\} $
also distinguishes the orthogonal density matrices $\left\{ \rho_{1},...,\rho_{N}\right\} $
follows from Lemma 1.
\end{proof}
The condition in Theorem 1, though remarkably simple and intuitive,
is able to capture the essence of local state discrimination and,
in particular, the role of entanglement therein. In particular, Theorem
1 leads to what we call {}``subspace degeneracy,'' which is discussed
in the next section.

\subsection{Subspace degeneracy}

As noted in the Introduction, intuitively, subspace degeneracy means
that any given set $\mathbb{S}$ of orthogonal density matrices belongs
to a collection of infinitely many sets having identical distinguishability
properties no matter how different the average entanglement of the
individual sets are. For a given set $\mathbb{S}=\left\{ \sigma_{i}:i=1,...,N\right\} $
whose local distinguishability is under consideration, consider another
orthogonal set $\mathbb{S}^{\prime}=\left\{ \sigma_{i}^{\prime}:i=1,...,N\right\} $
with the property that for every $i$, $\sigma_{i}$ and $\sigma_{i}^{\prime}$
have identical support. Let $\mathbb{Q}=\left\{ \mathbb{S}^{\prime}\right\} $
be the collection of all such orthogonal sets $\mathbb{S}^{\prime}$.
Clearly, $\mathbb{S}$ is also a member of $\mathbb{Q}$. In other
words, given a set of orthogonal subspaces $S=\left\{ s_{i}:i=1,...,N\right\} $,
$\mathbb{Q}$ is simply the collection of only those sets $\mathbb{S}^{\prime}=\left\{ \sigma_{i}^{\prime}:i=1,...,N\right\} $
with the properties that the for every $i$, $\sigma_{i}^{\prime}\in s_{i}$
and $\mbox{rank}\left(\sigma_{i}^{\prime}\right)=\dim s_{i}$. By
simple application of Theorem 1 we obtain the next result. 

\begin{prop} All orthogonal sets in $\mathbb{Q}$ are either perfectly
LOCC distinguishable or none of them is, regardless of the average
entanglement of the individual sets. Furthermore, if the sets can
be perfectly distinguished by LOCC, then there is a separable measurement
$\Pi_{\mathbb{Q}}$ that distinguishes every set in $\mathbb{Q}$.
\end{prop}

Simply put, no matter how different the average entanglement of the
sets might be, as far as perfect local distinguishability is concerned
they are either equally hard or equally easy to distinguish. This
is in sharp contrast to pure states, where the result in \cite{Hayashi-etal-2006}
implies different upper bounds for pure ensembles of the same cardinality
but having different average entanglement. Thus, unlike pure states,
there cannot be any direct correlation between entanglement (under
any reasonable measure) of the states and their local distinguishability.
Furthermore, entanglement or mixedness of the states in $\mathbb{S}$
is not be crucial in determining whether $\mathbb{S}$ can be perfectly
distinguished or not by LOCC. We use this fact (Proposition 3) to
obtain two kinds of upper bounds on the number of perfectly LOCC distinguishable
orthogonal states.

\subsection{Upper bounds on the number of perfectly LOCC distinguishable orthogonal
states}

In this section we give two kinds of upper bounds on the number of
perfectly LOCC distinguishable density matrices. In the first one,
the bounding quantities depend only on the maximal pure-state entanglement
in the supports of the density matrices, whereas in the second we
use Proposition 3 to optimize the bounding quantities over all sets
in $\mathbb{Q}$. 

We first show how local distinguishability of a set of orthogonal
density matrices can be related to the local distinguishability of
a set of orthogonal pure states satisfying certain conditions. For
the orthogonal density matrices $\sigma_{1},...,\sigma_{N}$, let
$S=\left\{ |\psi_{1}\rangle,|\psi_{2}\rangle,...,|\psi_{N}\rangle\right\} $
be a collection of orthogonal pure states such that for every $i$,
$|\psi_{i}\rangle\in s_{i}$, where, $s_{i}$ is the support of $\sigma_{i}$. 

\begin{prop} If the states $|\psi_{1}\rangle,|\psi_{2}\rangle,...,|\psi_{N}\rangle$
are not perfectly distinguishable by LOCC then the density matrices
$\sigma_{1},...,\sigma_{N}$ cannot be perfectly distinguished by
LOCC. \end{prop}
\begin{proof}
The proof of the second statement is by contradiction. Suppose states
$|\psi_{1}\rangle,|\psi_{2}\rangle,...,|\psi_{N}\rangle$ cannot be
perfectly distinguished by LOCC but there is a LOCC protocol that
perfectly distinguishes the density matrices $\sigma_{1},...,\sigma_{N}$.
From Theorem 1 we know that if the density matrices $\sigma_{1},...,\sigma_{N}$
are perfectly LOCC distinguishable then one can also perfectly distinguish
the orthogonal subspaces $s_{1},s_{2},...,s_{N}$, where $s_{i}$
is the support of $\sigma_{i}$. This implies, by Lemma 1, that the
states $|\psi_{1}\rangle,|\psi_{2}\rangle,...,|\psi_{N}\rangle$ can
also be perfectly distinguished locally because for every $i$, $|\psi_{i}\rangle\in s_{i}$,
which contradicts our assumption. 
\end{proof}
It is important to note that if the states $|\psi_{1\rangle},|\psi_{2}\rangle,...,|\psi_{N}\rangle$
can be perfectly distinguished locally then it does \emph{not} mean
that the density matrices $\sigma_{1},...,\sigma_{N}$ can also be
reliably distinguished. For example, consider the following density
matrices in $2\otimes2$: $\sigma_{1}=\alpha|\Phi^{+}\rangle\langle\Phi^{+}|+\left(1-\alpha\right)|01\rangle\langle01|$
and $\sigma_{2}=\beta|\Phi^{-}\rangle\langle\Phi^{-}|+\left(1-\beta\right)|10\rangle\langle10|$
, where $\alpha\neq0$ and $\beta\neq0$. Clearly, $s_{1}=\mbox{span}\left\{ |\Phi^{+}\rangle,|01\rangle\right\} $
and $s_{2}=\mbox{span}\left\{ |\Phi^{-}\rangle,|10\rangle\right\} $.
While any two orthogonal vectors $|\psi_{1}\rangle,|\psi_{2}\rangle$,
where $|\psi_{i}\rangle\in s_{i}$ for $i=1,2$, are perfectly LOCC
distinguishable, the density matrices $\sigma_{1}$ and $\sigma_{2}$
are not. The reason is that neither of the subspaces can be spanned
only by product states thereby violating a necessary condition for
perfect local discrimination by separable measurements \cite{Duan2007-2009}.

To arrive at our upper bound we will use the previous proposition
and inequality (\ref{eq:Hayashi}). For a set of orthogonal pure states
$\left\{ |\phi_{1}\rangle,|\phi_{2}\rangle,...,|\phi_{N}\rangle\right\} $
inequality (\ref{eq:Hayashi}) becomes, \begin{equation}
N\leq\frac{D}{\overline{1+R(|\phi_{i}\rangle)}}\leq\frac{D}{\overline{2^{E_{R}(|\phi_{i}\rangle)}}}\leq\frac{D}{\overline{2^{E_{g}(|\phi_{i}\rangle)}}},\label{hayashi-pure}\end{equation}
where, the corresponding bounding quantities have been defined before. 

Now, if the states $|\psi_{1}\rangle,|\psi_{2}\rangle,...,|\psi_{N}\rangle$
as defined in Proposition 4 violate the above inequality then we can
certainly conclude that the density matrices $\sigma_{1},...,\sigma_{N}$
are not perfectly distinguishable by LOCC. Therefore, if the density
matrices $\sigma_{1},\sigma_{2},...,\sigma_{N}$ are perfectly LOCC
distinguishable, then the following inequality holds: \begin{equation}
N\leq\frac{D}{\overline{1+R(|\psi_{i}\rangle)}}\leq\frac{D}{\overline{2^{E_{R}(|\psi_{i}\rangle)}}}\leq\frac{D}{\overline{2^{E_{g}(|\psi_{i}\rangle)}}}.\label{hayashi-1}\end{equation}
Note that the inequality may still be satisfied even if the density
matrices are locally indistinguishable. The example given after Proposition
4 conforms to this fact. The crucial point is that, \emph{if the density
matrices are locally distinguishable then the inequality will not
be violated. }

Naturally we would like to maximize the bounding quantities over all
orthogonal pure state ensembles like $S$. Let $S_{\max}=\left\{ |\Psi_{1}\rangle,|\Psi_{2}\rangle,...,|\Psi_{N}\rangle\right\} $
be the set of orthogonal pure states with the properties that for
every $i$, (a) $|\Psi_{i}\rangle\in s_{i}$, and (b) $R\left(|\Psi_{i}\rangle\right)=\max_{\psi\in s_{i}}R\left(|\psi\rangle\right)$.
The first condition ensures that the states belong to the supports
of the density matrices, so that proposition 4 is applicable. The
second condition reflects the fact that for every $i$, $|\Psi_{i}\rangle$
is the state with maximum pure state entanglement in the support of
$\sigma_{i}$. Thus by replacing the pure state ensemble $S$ by $S_{\max}$
in (\ref{hayashi-1}) we have the following result. 

\begin{thm} If the density matrices $\sigma_{1},\sigma_{2},...,\sigma_{N}$
are perfectly LOCC distinguishable, then, \begin{equation}
N\leq\frac{D}{\overline{1+R(|\Psi_{i}\rangle)}}\leq\frac{D}{\overline{2^{E_{R}(|\Psi_{i}\rangle)}}}\leq\frac{D}{\overline{2^{E_{g}(|\Psi_{i}\rangle)}}}\label{hayashi-2}\end{equation}
where, for every $i$, $|\Psi_{i}\rangle\in s_{i}$, $R\left(\Psi_{i}\right)=\max_{\psi\in s_{i}}R\left(|\psi\rangle\right)$.
\end{thm} 

Let us note that an exact analytical formula for robustness $R$ is
known for pure bipartite states \cite{Harrow-3-Steiner-Vidal}. Therefore,
for any given set of bipartite orthogonal density matrices, the upper
bound can be explicitly calculated (one needs to optimize to get the
best possible bound). Inequality (\ref{hayashi-2}) shows that mixed
states also admit pure-state-like correlation between entanglement
and the number of locally distinguishable states. This allows us to
make a general statement on the connection between entanglement and
local distinguishability: \emph{The number of perfectly LOCC distinguishable
quantum states, pure or mixed, is bounded above by the total dimension
over the average of maximal pure state entanglement in the supports
of the states. }It is, however, important to note that the quantity,
second from left, in inequality (\ref{eq:Hayashi}) is always stronger
than the leftmost quantity in (\ref{hayashi-2}) \cite{Virmani-Pvt}.

Our second bound can be considered to be the optimized version of
the general mixed-state bound given by (\ref{eq:Hayashi}). From Proposition
3 we know that if the set $\mathbb{S}$ is perfectly LOCC distinguishable
then so is any set $\mathbb{S}^{\prime}\in\mathbb{Q}$, and the measurement
that distinguishes $\mathbb{S}$ also distinguishes any $\mathbb{S}^{\prime}$
and vice versa. Noting that the sets are of same cardinality, it simply
follows that an upper bound on the number of perfectly LOCC distinguishable
states for any $\mathbb{S}^{\prime}$ is also an upper bound for $\mathbb{S}$.
The optimal bound is thus obtained by maximizing the bounding quantities
over all $\mathbb{S}^{\prime}$. 

For the given set $\mathbb{S}=\left\{ \sigma_{i};i=1,...,N\right\} $,
let $Q_{i}$ be the set of all density matrices in $s_{i}$ having
rank equal to $\dim s_{i}$, where $s_{i}$ is the support of $\sigma_{i}$.
Define the following quantities: \begin{eqnarray*}
\mathcal{R}_{i} & = & \max_{\sigma^{\prime\prime}\in Q_{i}}\mathcal{R}\left(\sigma^{\prime\prime}\right),\\
\mathcal{E}_{i} & = & \max_{\sigma^{\prime\prime}\in Q_{i}}\left(E_{R}\left(\sigma^{\prime\prime}\right)+S(\sigma^{\prime\prime})\right),\\
\mathcal{G}_{i} & = & \max_{\sigma^{\prime\prime}\in Q_{i}}G\left(\sigma^{\prime\prime}\right),\end{eqnarray*}
where $\mathcal{R}\left(\sigma^{\prime\prime}\right):=\alpha^{-1}\left(1+R_{g}\left(\sigma^{\prime\prime}\right)\right)$,
with $\alpha$ being the maximum eigenvalue of $\sigma^{\prime\prime}$,
$R_{g}\left(\sigma^{\prime\prime}\right)$ is the \emph{global} robustness
of entanglement \cite{Harrow-3-Steiner-Vidal}, $E_{R}\left(\sigma^{\prime\prime}\right)$
is the relative entropy \cite{Vedral-Plenio-1998}, $S\left(\sigma^{\prime\prime}\right)$
is the von Neumann entropy, and $G\left(\sigma^{\prime\prime}\right)$
is the geometric measure \cite{Hayashi-etal-2006}.

\begin{thm} If the set of states $\mathbb{S}=\left\{ \sigma_{i};\, i=1,...,N\right\} $
is perfectly distinguishable by LOCC, then the number of states is
bounded by\begin{equation}
N\leq D/\overline{d\left(\sigma_{i}\right)}\leq D/\overline{\mathcal{R}_{i}}\leq D/\overline{2^{\mathcal{E}_{i}}}\leq D/\overline{2^{\mathcal{G}_{i}}},\label{Hayashi-revised}\end{equation}
 where, $\overline{x_{i}}=\frac{1}{N}\sum_{i=1}^{N}x_{i}$ denotes
the average. \end{thm}
\begin{proof}
It is sufficient to show that \begin{equation}
d\left(\sigma_{i}\right)\geq\mathcal{R}_{i}\geq\mathcal{E}_{i}\geq\mathcal{G}_{i}.\label{thm-3-eq-1}\end{equation}
 Inequality (\ref{Hayashi-revised}) is then obtained by combining
Proposition 2 and the above inequality and dividing by $N$. The proof
follows along lines very similar to that in \cite{Hayashi-etal-2006}. 

It was shown in \cite{Hayashi-etal-2006} that we can write $d\left(\sigma_{i}\right)=\min\left(\frac{1}{\lambda}\right)$
such that $\exists$ $\varrho^{\prime}$, satisfying\begin{equation}
\tilde{\Pi_{i}}=\lambda_{i}\mathcal{P}_{i}+\left(1-\lambda_{i}|\mathcal{P}_{i}|\right)\varrho^{\prime}\in\mathfrak{S}\label{thm-3-eq-4}\end{equation}
 along with the conditions\begin{eqnarray}
\mbox{Tr}\left(\mathcal{P}_{i}\varrho^{\prime}\right) & = & 0,\\
\langle\phi|\tilde{\Pi_{i}}|\phi\rangle & \geq & \lambda\;\forall|\phi\rangle,\end{eqnarray}
where $\tilde{\Pi_{i}}=\Pi_{i}/\mbox{Tr}\left(\Pi_{i}\right)$, $\mathcal{P}_{i}$
is the projector onto $s_{i}$ (the support of $\sigma_{i}$) , $|\mathcal{P}_{i}|=\dim s_{i}$,
and $\mathfrak{S}$ is the set of separable density matrices. Recall
that $Q_{i}$ is the set of all density matrices in $s_{i}$ having
rank equal to $|\mathcal{P}_{i}|$. Now observe that for any density
matrix $\sigma^{\prime\prime}\in Q_{i}$, we have $d\left(\sigma^{\prime\prime}\right)=d\left(\sigma_{i}\right)$
\cite{th-3-proof}. Now $\sigma^{\prime\prime}$ can be expressed
as \begin{equation}
\sigma^{\prime\prime}=\alpha\mathcal{P}_{i}-\beta\sigma^{\prime\prime\prime},\label{thm-3-eq-5}\end{equation}
 where $\alpha$ is the maximum eigenvalue of $\sigma^{\prime\prime}$,
$\beta=|\mathcal{P}_{i}|\alpha-1$, and $\sigma^{\prime\prime\prime}\in s_{i}$.
Equation (\ref{thm-3-eq-4}) can therefore be rewritten in the form,
\begin{equation}
\tilde{\Pi_{i}}=\frac{\lambda_{i}}{\alpha}\left(\sigma^{\prime\prime}+\gamma\varrho^{\prime\prime}\right)\in\mathfrak{S}\label{thm-3-eq-6}\end{equation}
 where, $\gamma=\left(\alpha-\lambda_{i}\right)/\lambda_{i}$. Noting
that the generalized (or global) robustness of entanglement \cite{Harrow-3-Steiner-Vidal}
of $R_{g}\left(\sigma\right)$ of any state $\sigma$ is defined by
$R_{g}\left(\sigma\right)=\min t$ such that there exists a state
$\varrho$ satisfying \begin{equation}
\frac{1}{1+t}\left(\sigma+t\varrho\right)\in\varsigma,\label{them-3--global-robust}\end{equation}
 where $\varsigma$ is a separable state, we immediately obtain $R_{g}\left(\sigma^{\prime\prime}\right)\leq\gamma$.
Thus for any $\sigma^{\prime\prime}\in Q_{i},$ we have \begin{equation}
d\left(\sigma^{\prime\prime}\right)\geq\alpha^{-1}\left[1+R_{g}\left(\sigma^{\prime\prime}\right)\right].\label{thm-3-eq-7}\end{equation}
 Thus, \begin{equation}
d\left(\sigma_{i}\right)\geq\mathcal{R}_{i}=\max_{\sigma^{\prime\prime}\in\mathcal{Q}_{i}}\mathcal{R}\left(\sigma^{\prime\prime}\right),\label{thm-3-eq-8}\end{equation}
 where $\mathcal{R}\left(\sigma^{\prime\prime}\right):=\alpha^{-1}\left[1+R_{g}\left(\sigma^{\prime\prime}\right)\right]$.
The rest of the proof is straightforward. It is easy to show that
for any density matrix $\sigma^{\prime\prime}\in Q_{i}$, the following
inequality holds: \begin{equation}
r\left(\sigma_{i}\right)\geq\mathcal{E}\left(\sigma^{\prime\prime}\right)\geq\mathcal{G}\left(\sigma^{\prime\prime}\right).\label{thm-3-eq-9}\end{equation}
 Thus we obtain \begin{equation}
\mathcal{R}_{i}\geq r\left(\sigma_{i}\right)\geq\mathcal{E}_{i}=\max_{\sigma^{\prime\prime}\in Q_{i}}\mathcal{E}\left(\sigma^{\prime\prime}\right)\geq\mathcal{G}_{i}=\max_{\sigma^{\prime\prime}\in Q_{i}}\mathcal{G}\left(\sigma^{\prime\prime}\right).\label{thm3-eq-10}\end{equation}
 Combining Eqs.$\,$ (\ref{thm3-eq-10}) and (\ref{thm-3-eq-8}),
we get Eq.$\,$(\ref{thm-3-eq-1}). This concludes the proof.
\end{proof}
A few remarks are in order. 

(i) By construction, for every bounding quantity, $\mathcal{R},$
$\mathcal{E}$, and $\mathcal{G}$, there always exists a set of orthogonal
quantum states $\mathbb{S}^{\prime}=\left\{ \sigma_{i}^{\prime};\, i=1,...,N\right\} \in\mathbb{Q}$
maximizing it, which is what the essence of the entire optimality
argument. For example, one can construct an orthogonal set $\mathbb{S}_{\mathcal{R}}^{\prime}\left(\sigma^{\prime}\right)=\left\{ \sigma_{i}^{\prime};\, i=1,...,N\right\} \in\mathbb{Q}$,
such that for every $i$, $\mathcal{R}_{i}=\mathcal{R}\left(\sigma_{i}^{\prime}\right)$,
and similarly for the quantities $\mathcal{E}$, and $\mathcal{G}$. 

(ii) A nice feature of the above inequality is that the hierarchical
form holds even when the bounding quantities are independently maximized
(this is clear from the proof), and different sets may maximize different
quantities. 

(iii) For any set $\mathbb{S}^{\prime\prime}=\left\{ \sigma_{i}^{\prime\prime};\, i=1,...,N\right\} \in\mathbb{Q}$,
the following inequality holds {[}inequality (\ref{Hayashi-revised})
is simply the optimized version of the following one{]}:

\begin{equation}
N\leq D/\overline{d\left(\sigma_{i}\right)}\leq D/\overline{\mathcal{R}\left(\sigma^{\prime\prime}\right)}\leq D/\overline{2^{\mathcal{E}\left(\sigma^{\prime\prime}\right)}}\leq D/\overline{2^{\mathcal{G}\left(\sigma^{\prime\prime}\right)}}.\label{Hayashi-3}\end{equation}

\section{Conclusions}

We have considered the problem of local distinguishability of orthogonal
mixed states. In particular, we have investigated how entanglement
and mixedness of the states influence their local distinguishability.
We have shown a general equivalence between local discrimination of
orthogonal states and subspaces which in turn implies that local distinguishability
of mixed states is completely determined by whether or not their supports
are also locally distinguishable. This led to the following results:
(a) state specific properties like inseparability, mixedness of the
density matrices do not have any special role in determining their
local distinguishability, (b) local distinguishability of mixed states
may be completely determined by maximal pure state entanglement within
their supports (c) an upper bound on the number of perfectly locally
distinguishable orthogonal mixed states is given where the bounding
quantities are optimized over all orthogonal mixed state ensembles
having identical supports. 

Although the results obtained in this paper and in \cite{Hayashi-etal-2006}
show that entanglement is a significant factor in local distinguishability,
many questions still remain open. For example, there are orthogonal
product states known to be locally indistinguishable \cite{Bennett-I-99,Bennett-II-99 +Divin-2003},
despite being completely unentangled. Whether there is a deeper reason
behind this phenomena or is it just a consequence of the fact that
not all separable measurements are locally implementable is not known
yet. Obviously entanglement of the states is a non-issue here, but
entanglement could still be important because to implement such separable
measurements by LOCC, one is expected to consume auxiliary entanglement.
Thus, it is necessary to quantify the entanglement cost of such separable
measurements. Another interesting class of states requiring further
investigation are those which despite being entangled and locally
indistinguishable, do not violate the inequalities presented in this
paper or in \cite{Hayashi-etal-2006}. All these examples show that
there is more to local distinguishability of quantum states than what
can be captured through entanglement only. \\

\textbf{Acknowledgments:} The author is grateful to Guruprasad Kar
(ISI, Kolkata) for many helpful discussions. Thanks to S. Virmani
and D. Markham for their comments on an earlier version of this manuscript.


\begin{thebibliography}{26}
\bibitem{support} The support of a density matrix is the subspace
spanned by its eigenvectors corresponding to the nonzero eigenvalues.
\bibitem{th-3-proof} Note that the definition of $d(\sigma)$ involves
minimization over all separable POVMs satisfying Eq.$\,$(\ref{prop1-eq-1}).
From theorem 2 we know that the POVM that distinguishes $\mathbb{S}$
also distinguishes any $\mathbb{S}^{\prime}\in\mathbb{Q}$.
\end{thebibliography}
\end{document}